\documentclass[10pt,conference]{IEEEtran}\IEEEoverridecommandlockouts

\usepackage{amsfonts}
\usepackage{slashbox}
\usepackage{amsmath}
\usepackage{multirow}
\usepackage{graphicx}
\usepackage{amsfonts}
\usepackage{amsmath,epsfig}
\usepackage{mathbbold}
\usepackage{stfloats}
\usepackage{amssymb}
\usepackage{cite}
\usepackage{amsmath}
\usepackage{amssymb}
\usepackage{latexsym}
\usepackage{multirow}
\usepackage{epsfig}
\usepackage{graphics}
\usepackage{mathrsfs}
\usepackage{algorithmic}
\usepackage{algorithm}
\usepackage{subfigure}
\usepackage{slashbox}
\usepackage[all]{xy}

\usepackage{pifont}
\usepackage{bbding}
\usepackage{stmaryrd}
\usepackage{amssymb}
\usepackage{amsfonts}
\usepackage{epic}
\usepackage{stfloats}
\usepackage{latexsym}
\usepackage{epstopdf}
\usepackage{epic}
\usepackage{bm}
\usepackage{xcolor}
\usepackage{mathrsfs}
\usepackage{pifont}
\usepackage{bbding}
\usepackage{amsmath,epsfig}
\usepackage{mathbbold}
\usepackage{stmaryrd}
\usepackage{amssymb}
\usepackage{amsfonts}
\usepackage{epic}
\usepackage{graphicx}
\usepackage{subfigure}

\usepackage{enumerate}

\usepackage{stfloats}
\usepackage{latexsym}
\usepackage{epstopdf}
\usepackage{epic}
\usepackage{multirow}
\usepackage{stfloats}
\usepackage{bm}

\usepackage{booktabs}
\usepackage{color}
\newtheorem{theorem}{\underline{Theorem}}

\makeatletter

\newcommand{\Rmnum}[1]{\expandafter\@slowromancap\romannumeral #1@}
\makeatother
\begin{document}
\title{Spectrum-Power Trading for  Energy-Efficient \\Small Cell }

\author{\IEEEauthorblockN{Qingqing Wu\IEEEauthorrefmark{1}\IEEEauthorrefmark{2}, Geoffrey Ye Li\IEEEauthorrefmark{1}, Wen Chen\IEEEauthorrefmark{2}, and Derrick Wing Kwan Ng\IEEEauthorrefmark{3}}
\IEEEauthorblockA{\IEEEauthorrefmark{1} School of Electrical and Computer Engineering, Georgia Institute of Technology, USA.\\}
\IEEEauthorblockA{\IEEEauthorrefmark{2}Department of Electronic Engineering, Shanghai Jiao Tong University, Shanghai, China.\\ }
\IEEEauthorblockA{\IEEEauthorrefmark{3}School of Electrical Engineering and Telecommunications, The University of New South Wales, Australia.\\
Emails: \{qingqing.wu, liye\}@ece.gatech.edu, wenchen@sjtu.edu.cn, w.k.ng@unsw.edu.au}}


\maketitle
\begin{abstract}
This paper investigates spectrum-power trading between a small cell (SC) and a macro-cell (MC), where the SC consumes power to serve the macro-cell users (MUs) in exchange for some bandwidth from the MC.
Our goal is to maximize the system energy efficiency (EE) of the SC while guaranteeing the quality of service (QoS) of each MU as well as small cell users (SUs). Specifically, given the minimum data rate requirement and the bandwidth provided by the MC, the SC jointly optimizes MU selection, bandwidth allocation, and power allocation while guaranteeing its own minimum required system data rate. The problem is challenging due to the binary MU selection variables and the fractional form objective function. We first show that in order to achieve the maximum system EE, the bandwidth of an MU is shared with at most one SU in the SC. Then, for a given MU selection, the optimal bandwidth and power allocations are obtained by exploiting the fractional programming. To perform MU selection, we first introduce the concept of trading EE.
Then, we reveal a sufficient and necessary condition for serving an MU without considering the total power constraint and the minimum data rate constraint.
 Based on this insight, we propose a low computational complexity MU selection algorithm.
 Simulation results demonstrate the effectiveness of the proposed algorithms.

\end{abstract}
\renewcommand{\baselinestretch}{0.96}
\normalsize
\section{Introduction}
The fifth generation (5G) mobile networks are expected to provide ubiquitous ultra-high data rate services and seamless user experience across the whole communication system \cite{hu2014energy}. The concept of small cell (SC) networks, such as femtocells, has been recognized as a key technology that can significantly enhance the performance of 5G networks. The underlay SCs enable the macro-cells (MCs) to offload huge volume of data and large numbers of users. In particular, the SC could help to serve some macro-cell users (MUs) with high required data rate, especially when these MUs are far away from the MC base station (BS), e.g. cell edge users.
 Although the MUs offloading reduces the power consumption of MCs, additional power consumption is imposed to SCs that may degrade the quality of services (QoS) of small cell users (SUs).  Therefore, motivating the SC  to serve MUs  is a critical problem,  especially when the SC BS does not belong to the same mobile operator with the MC BS.

 In recent years, the explosive growth of data hungry applications and various services has triggered a dramatic increase in energy consumption of wireless communication systems. Due to rapidly rising energy costs and tremendous carbon footprints \cite{ng2012energy1}, energy efficiency (EE), measured in bits-per-joule, has attracted considerable attention as a new performance metric in both academia and industry \cite{ng2012energy1,sun2013energy,liu15,yiran2015,
  ramamonjison2015energy,han2014spectrum}. For instance,
 energy-efficient resource allocation was studied in \cite{ng2012energy1} for  single-cell systems with a large number of base station antennas.
Subsequently, EE maximization problems are further investigated for various practical scenarios such as  relaying systems \cite{sun2013energy},   full duplex communications \cite{liu15}, heterogenous networks \cite{yiran2015}, cognitive radio (CR) networks \cite{ramamonjison2015energy}, coordinated multi-point
(CoMP) transmission \cite{han2014spectrum}, etc.
However, all these previous works do not take into account spectrum sharing or energy cooperation between the SC and the MC, which are expected to enhance the performances of both networks simultaneously.

  In this paper,  we study spectrum-power trading between an SC and an MC where the SC BS consumes additional power to serve MUs while the MC allows the SC BS to operate on some bandwidth of the MC.  
To enable energy-efficient SC via spectrum-power trading, we need to address the following fundamental issues. First, when should the SC serve an MU?  For example,  if the required data rate of an MU is  too stringent but the bandwidth assigned to it is insufficient, it may not be beneficial for the SC to serve that MU.
 Second, how much bandwidth should be obtained and how much power should be utilized in order to achieve the maximum EE as well as guaranteeing the QoS of the MUs? This question naturally arises because if the SC desires to acquire more bandwidth from the MU, it has to transmit a higher transmit power to serve this MU.
 However, this may in turn leave a lower transmit power budget for its own SUs and thereby lead to a lower system date rate as well as an unsatisfactory system EE. Thus, there exists a non-trivial spectrum and power tradeoff in the spectrum-power trading.
 These issues are critical but have not been investigated in previous works   \cite{li2014energy,ng2012energy1,liu15,sun2013energy,
yiran2015,ramamonjison2015energy}, yet, and we will address them in this paper.
\section{System Model}

\subsection{Spectrum-Power Trading Model between SC and MC}

 We consider a spectrum-power trading scenario which consists of an MC and an SC, as depicted in Figure \ref{system_model}. The MC BS aims at offloading the data traffic of some cell edge MUs to the SC BS in order to reduce its own power consumption.  The set of MUs potentially offloaded to the SC is denoted by $\mathcal{K}$ with $|\mathcal{K}|=K$ and the set of  SUs in the SC is denoted by $\mathcal{N}$ with $|\mathcal{N}|=N$, where $|\cdot|$ indicates the cardinality of a set. Each MU and SU have been assigned with a piece of licensed bandwidth by the MC and the SC, respectively, denoted as $W^k_{\mathrm{MC}}$ and $B^n_{\mathrm{SC}}$. 
To simplify the problem,  we assume that the SC BS as well  as each user is equipped with a single-antenna \cite{guo2014joint}. Besides, typical optical fiber with high capacity can be deployed to connect the SC and the MC for data offloading purpose.


The channels between the SC and MUs as well as SUs are assumed to be quasi-static block fading \cite{ng2012energy1}. We also assume that SU $n$, $\forall\, n\in \mathcal{N}$, experiences frequency flat fading\footnote{The current work can be easily extended to the case of frequency selective fading at the expense of a more involved notation.} on its own licensed bandwidth $B^n_{\mathrm{SC}}$ and each MU $k$'s bandwidth $W^k_{\mathrm{MC}}$, respectively. In addition,  MU $k$, $\forall\, k\in \mathcal{K}$, also experiences frequency flat fading on its own licensed bandwidth $W^k_{\mathrm{MC}}$.
It is also assumed that the channel state information (CSI) of all the users is perfectly known at the SC in order to explore the EE upper bound and extract useful design insights of the considered systems.

\begin{figure}[!t]
\centering
\includegraphics[width=3.5in]{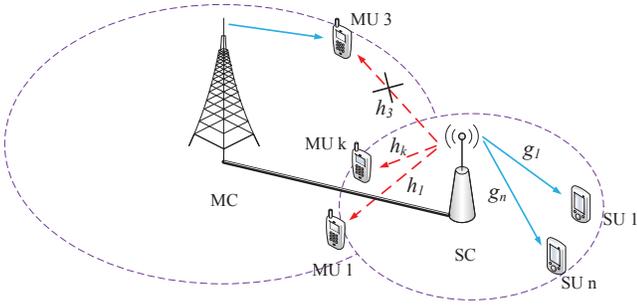}
\caption{The spectrum-power trading model between an SC and an MC.  For example, the SC may agree to serve MU $1$ but refuse to serve MU $3$ in order to maximize its system EE. \label{system_model}}
\end{figure}

For MU $k$, $\forall\, k\in \mathcal{K}$, the channel power gain between the SC  and MU $k$ on its own licensed bandwidth $W_{\mathrm{MC}}^k$ is denoted as $h_k$, cf. Figure \ref{system_model}. The corresponding transmit power and the bandwidth allocated to MU $k$ by the SC are denoted as $q_{k}$ and $w_k$, respectively.
Thus, the achievable data rate of MU $k$ can be expressed as
\begin{align}
r_k=w_k\log_2\left(1+\frac{q_kh_k}{w_kN_0}\right),
\end{align}
where $N_0$ is the spectral density of the additive white Gaussian noise.

For SU $n$, $\forall\, n\in \mathcal{N}$, the channel power gain between the SC  and SU $n$ on its own licensed bandwidth $B^n_{\mathrm{SC}}$ is denoted as $g_{n}$, cf. Figure \ref{system_model}. The corresponding transmit power is denoted as $p_n$. Then, the achievable date rate of SU $n$ on its own bandwidth can be expressed as
\begin{align}
r^n_{\mathrm{SC}}=B^n_{\mathrm{SC}}\log_2\left(1+\frac{p_{n}g_{n}}{B^n_{\mathrm{SC}}N_0}\right).
\end{align}
In addition to $B^n_{\mathrm{SC}}$, each SU may acquire some additional bandwidth from MUs due to the proposed spectrum-power trading between the SC and the MC.
Denote the channel power gain between the SC and SU $n$ on the bandwidth of MU $k$ as $g_{k,n}$. The bandwidth that the SC allocates to SU $n$ from $W^k_{\mathrm{MC}}$ is denoted as $b_{k,n}$ and the corresponding transmit power is denoted as $p_{k,n}$.
 Then, the achievable data rate of SU $n$ on the bandwidth of MU $k$ can be expressed as
\begin{align}
r_{k,n} = b_{k,n} \log_2\left(1+\frac{p_{k,n}g_{k,n}}{b_{k,n}N_0}\right).
\end{align}
Thus, the total data rate of SU $n$ under the proposed spectrum-power trading is given by
\begin{align}
R_n=r^n_{\mathrm{SC}} + \sum_{k=1}^Kx_kr_{k,n},
\end{align}
where $x_k$ is the MU selection binary variable and defined as
\begin{equation}\label{eq5}
x_k=\left\{
\begin{array}{lcl}
1,& \text{if MU $k$ is served by the SC},\\
0,& \text{otherwise}.
\end{array}\right.
\end{equation}
Therefore, the overall system data rate of SUs is expressed as
\begin{align}\label{eq6}
R_{\rm{tot}} =  \sum_{n=1}^{N}R_n = \sum_{n=1}^Nr^n_{\mathrm{SC}} + \sum_{n=1}^N\sum_{k=1}^Kx_kr_{k,n}.
\end{align}

\subsection{Power Consumption Model for SC BS}
 Here, we adopt the power consumption model from \cite{ng2012energy1} in which the overall energy consumption of the SC BS consists of two parts: the dynamic power consumed in the power amplifier for transmission, $P_{{t}}$, and the static power consumed for circuits, $P_\mathrm{c}$.

  The dynamic power consumption is modeled as a linear function of the transmit power that includes both the transmit power consumption for SUs and that for MUs, i.e.,
 \begin{align}\label{eq7}
P_{{t}}=\sum_{n=1}^{N}\frac{p_{n}}{\xi} + \sum_{n=1}^{N}\sum_{k=1}^{K}x_k\frac{p_{k,n}}{\xi}+\sum_{k=1}^{K}x_k\frac{q_{k}}{\xi},
\end{align}
where $\xi\in (0,1]$  is a constant that accounts for the power amplifier (PA) efficiency and the value of  $\xi$ depends on the specific type of the BS PA.
The static power consumption for circuits is denoted as  $P_\mathrm{c}$, which is caused by filters, frequency synthesizers, etc.  Therefore, the overall energy consumption of the SC BS can be expressed as
\begin{align}\label{eq7}
P_{\rm{tot}}=\sum_{n=1}^{N}\frac{p_{n}}{\xi} + \sum_{n=1}^{N}\sum_{k=1}^{K}x_k\frac{p_{k,n}}{\xi}+\sum_{k=1}^{K}x_k\frac{q_{k}}{\xi}+P_\mathrm{c}.
\end{align}

\section{Problem Formulation and Analysis}
    Our goal is to enhance the system EE of the SC via spectrum-power trading while guaranteing the QoS of the MUs as well as the SC network. Thus, the system EE of the SC is defined as the ratio of the total achievable data rate of SUs to the total power consumption that includes not only the power consumed for providing services for SUs, but also the power consumed for spectrum-power trading, i.e., $EE=\frac{R_{\rm{tot}}}{P_{\rm{tot}}}$.
Specifically, we aim to maximize the system EE of the SC via jointly optimizing MU selection, bandwidth allocation, and power allocation.
 Let $\mathcal{S}=\Big\{\{{p}_{n}\}, \{{p}_{k,n}\}, \{{q}_k\}, \{{b}_{k,n}\}, \{{w}_k\}\Big\}$ denote the resource allocation solution. The system EE maximization problem is formulated as:
\begin{align}\label{eq16}
\mathop {\text{max} }\limits_{S} ~& \frac{\sum_{n=1}^Nr^n_{\mathrm{SC}} + \sum_{k=1}^K\sum_{n=1}^Nx_kr_{k,n}}
{\sum_{n=1}^{N}\frac{p_{n}}{\xi} + \sum_{n=1}^{N}\sum_{k=1}^{K}x_k\frac{p_{k,n}}{\xi}+\sum_{k=1}^{K}x_k\frac{q_{k}}{\xi}+P_\mathrm{c}}\nonumber \\
\text{s.t.} ~~~&  \text{C1:}~~\sum_{n=1}^{N}{p_{n}} +  \sum_{n=1}^N\sum_{k=1}^Kp_{k,n}+\sum_{k=1}^{K}q_k\leq P^{\mathrm{SC}}_{\mathop{\max}}, \nonumber \\
& \text{C2:}~~\sum_{n=1}^Nb_{k,n}+w_k\leq x_kW^k_{\mathrm{MC}},  ~\forall\, k\in\mathcal{ K}, \nonumber \\
&\text{C3:}~~w_{k}\log_2\left(1+\frac{q_{k}h_{k}}{w_{k}N_0}\right)\geq x_kR^k_{\mathrm{MC}}, ~\forall\, k\in\mathcal{ K},  \nonumber \\
& \text{C4:}~\sum_{n=1}^Nr^n_{\mathrm{SC}} + \sum_{k=1}^K\sum_{n=1}^Nx_kr_{k,n} \geq R^{\mathrm{SC}}_{\mathop{\min}}, \nonumber\\
& \text{C5:}~x_k\in\{0, 1\}, ~\forall\,  k\in\mathcal{ K},    \nonumber\\
& \text{C6:}~b_{k,n}\geq  0, ~  w_{k}\geq  0, ~\forall\, k\in\mathcal{ K}, n\in\mathcal{N},    \nonumber\\
& \text{C7:}~p_n\geq0,~ p_{k,n}\geq0, ~  q_{k}\geq0, ~\forall\,  k\in\mathcal{ K}, n\in\mathcal{N}.
\end{align}
In problem (\ref{eq16}),  C1 limits the maximum transmit power of the SC BS to $P^{\mathrm{SC}}_{\mathop{\max}}$.
C2 ensures that the bandwidth allocated to SUs and MU $k$ does not exceed the available bandwidth, $W^k_{\mathrm{MC}}$, that has been licensed to MU $k$ by the MC.
In C3, $R^k_{\mathrm{MC}}$ is the minimum data rate requirement of MU $k$.  C4 guarantees the minimum required system data rate of the SC.
C5 indicates whether to serve MU $k$ or not. Note that if $x_k=0$, then from C2 and C4, both $b_{k,n}$ and $q_k$ will be forced to zeros at the optimal solution of problem  (\ref{eq16}). In other words, the SC does not obtain additional bandwidth from MU $k$ and does not serve MU $k$ either. Therefore, the SC only performs spectrum-power trading with the MC if it is beneficial to the EE of the SC.
 C6 and C7 are non-negativity constraints on the bandwidth and power allocation variables, respectively. 

Note that problem (\ref{eq16}) is neither a concave nor a quasi-concave optimization problem due to the fractional form objective function and the binary optimization variables $x_k, \forall\, k$. 
Nevertheless, in the following theorem, we first transform the energy-efficient optimization problem into a simplified one based on its special structure.

\begin{theorem}
The optimal solution of problem (\ref{eq16}) is equivalent to that of the following problem:
\begin{align}\label{eq17}
\mathop {\text{max} }\limits_{S} ~& \frac{\sum_{n=1}^Nr^n_{\mathrm{SC}} + \sum_{k=1}^Kx_kr_{k,k'}}
{\sum_{n=1}^{N}\frac{p_{n}}{\xi} + \sum_{k=1}^{K}x_k\frac{p_{k,k'}}{\xi}+\sum_{k=1}^{K}x_k\frac{q_{k}}{\xi}+P_\mathrm{c}}\nonumber \\
\text{s.t.} ~~&  \text{C5}, ~ \text{C6},~ \text{C7},~  \nonumber \\
& \text{C1:}~ \sum_{n=1}^{N}{p_{n}} + \sum_{k=1}^Kp_{k,k'}+\sum_{k=1}^{K}q_k\leq P^{\mathrm{SC}}_{\mathop{\max}}, \nonumber \\
& \text{C2:}~b_{k,k'}+w_k= x_kW^k_{\mathrm{MC}},  ~\forall\, k\in\mathcal{ K}, \nonumber \\
&\text{C3:}~w_{k}\log_2\left(1+\frac{q_{k}h_{k}}{w_{k}N_0}\right)= x_kR^k_{\mathrm{MC}}, ~\forall\, k\in\mathcal{ K},  \nonumber \\
& \text{C4:}~\sum_{n=1}^Nr^n_{\mathrm{SC}} + \sum_{k=1}^Kx_kr_{k,k'} \geq R^{\mathrm{SC}}_{\mathop{\min}}, 
\end{align}
where $k'=\arg\mathop {\max }\limits_{n \in \mathcal{N}} ~g_{k,n}.$
\end{theorem}
\begin{proof}
Due to page limitation, we only provide a sketch of the proof. It can be shown that assigning the bandwidth  obtained from MU $k$ to SU $k'$, where $k'=\arg\mathop {\max }\limits_{n \in \mathcal{N}} ~g_{k,n}$, is able to achieve the maximum data rate for a given transmit power budget. In addition, C2 and C3 are satisfied with equalities at the optimal solution, since the objective function is an increasing function of $b_{k, k'}$ and an decreasing function of $q_k$. Thus, problem (\ref{eq16}) is simplified to problem (\ref{eq17}).
\end{proof}

Theorem $1$ suggests that if the SC decides to serve MU $k$, the most energy-efficient strategy is only to share the bandwidth of MU $k$ with at most one SU who has the largest channel power gain on the traded bandwidth, $W^k_{\mathrm{MC}}$. In addition, constraints C2 and C3 are also met with equalities at the optimal solution since it is always beneficial for the SC to seek as much as bandwidth while consuming as less as transmit power in the spectrum-power trading with the MC.
Although problem (\ref{eq17}) is more tractable than problem  (\ref{eq16}), it is still a combinatorial non-convex optimization problem. In general, there is no efficient method for solving non-convex optimization problems and performing exhaustive search among all the possible cases to find a globally optimal solution may lead to an exponential computational complexity, which is prohibitive in practice. Thus,  we aim to develop a low computational complexity approach via exploiting the special structure of the problem.

\section{Energy-Efficient Resource Allocation For Given MU Selection}
Denote $\Psi$ as a set of MUs that are scheduled by the SC, i.e., $\Psi\triangleq\{k\,|\,x_k=1, k\in \mathcal{K}\}$, and denote $EE_{\Psi}$ as the maximum system EE of problem (\ref{eq17}) based on set $\Psi$, i.e., $EE=EE_{\Psi}$. 
For a given $\Psi$, problem (\ref{eq17}) is reduced to a joint bandwidth and power allocation problem.
However, the reduced problem is still non-convex due to the fractional-form objective function.
In the following, we show that the optimal solution of the reduced problem can be efficiently obtained by exploiting the fractional structure of the objective function in (\ref{eq17}).
\subsection{Problem Transformation}
From the nonlinear fractional programming theory \cite{dinkelbach1967nonlinear}, for a problem of the form,
\begin{equation}\label{eq11c}
q^*= \mathop {\text{max} }\limits_{S\in \mathcal{F}} \frac{R_{\rm{\rm{tot}}}(S)}{P_{\rm{\rm{tot}}}(S)},
\end{equation}
where  $\mathcal{F}$ is the feasible solution set,
there exists an equivalent problem in  subtractive form that satisfies
\begin{equation}\label{eq12c}
T(q^*)=\mathop {\text{max} } \limits_{S\in \mathcal{F}}\Big\{R_{\rm{\rm{tot}}}(S)-q^*P_{\rm{\rm{tot}}}(S)\Big\}=0.
\end{equation}
The equivalence between (\ref{eq11c}) and (\ref{eq12c}) can be easily verified with the corresponding maximum value $q^*$ that is also the maximum system EE. Besides, Dinkelbach
 provides an iterative  method in \cite{dinkelbach1967nonlinear} to obtain  $q^*$.
  By applying this transformation to (\ref{eq17}) with  $b_{k,k'} = W^k_{\mathrm{MC}}- w_k$ and $q_{k}=\left(2^\frac{R^{k}_{\mathrm{MC}}}{w_k}-1\right)\frac{w_{k}N_0}{h_k}, \forall \, k \in \Psi$, we obtain the following optimization problem for a given $q$ in each iteration:
\begin{align}\label{eq13c1}
&\mathop {\text{max} }\limits_{\overset{\{p_{n}\}, \{p_{k,k'}\},}{ \{w_k\}}}~ \sum_{n=1}^NB^n_{\mathrm{SC}}\log_2\left(1+\frac{p_{n}g_{n}}{B^n_{\mathrm{SC}}N_0}\right) \nonumber \\
&+ \sum_{k\in\Psi}(W^k_{\mathrm{MC}}-w_k)\log_2\left(1+\frac{p_{k,k'}g_{k,k'}}{(W^k_{\max}-w_k)N_0}\right) \nonumber\\
&-q\left(\sum_{n=1}^{N}\frac{p_{n}}{\xi}+ \sum_{k\in \Psi}\frac{p_{k,k'}}{\xi}+\sum_{k\in \Psi}\left(2^\frac{R^{k}_{\mathrm{MC}}}{w_k}-1\right)\frac{w_{k}N_0}{\xi h_k}+P_\mathrm{c}\right)\nonumber\\
&\text{s.t.} ~  \text{C4:}~\sum_{n=1}^Nr^n_{\mathrm{SC}} + \sum_{k\in \Psi}r_{k,k'} \geq R^{\mathrm{SC}}_{\mathop{\min}}, ~ \text{C6:}~ w_{k}\geq  0, ~\forall\, k\in \Psi,    \nonumber\\ &\text{C1:}~ \sum_{n=1}^{N}{p_{n}} + \sum_{k\in \Psi}p_{k,k'}+\sum_{k\in \Psi}\left(2^\frac{R^{k}_{\mathrm{MC}}}{w_k}-1\right)\frac{w_{k}N_0}{h_k}\leq P^{\mathrm{SC}}_{\mathop{\max}}, \nonumber \\
& \text{C7:}~p_n\geq0, ~p_{k,k'}\geq0, ~\forall\,  k\in\Psi, n\in\mathcal{N}.
\end{align}
After the transformation, it can be verified that problem (\ref{eq13c1}) is jointly concave with respect to all the optimization variables and also satisfies Slater's constraint qualification \cite{ng2012energy1}.  As a result, the duality gap between problem (\ref{eq13c1}) and its dual problem is zero. Therefore, the optimal solution of problem (\ref{eq13c1}) can be obtained by applying the Lagrange duality theory \cite{ng2012energy1}. In the next section, we derive the optimal bandwidth and power allocation via exploiting the Karush-Kuhn-Tucker (KKT) conditions of problem (\ref{eq13c1}) that leads to a computationally efficient resource allocation algorithm.

\subsection{Joint Bandwidth and Power Allocation}
  Denote $\lambda$ and  $  \mu$ as the non-negative Lagrange multipliers associated with  constraints C1 and C4, respectively.
Then, the optimal bandwidth and power allocation be obtained as in Theorem \ref{theorem70}.
 \begin{theorem}\label{theorem70}
Given  $\lambda$ and ${\mu}$,  the optimal bandwidth and power allocation is given by
\begin{align}
w_{k}&= \min\left(\frac{R^k_{\mathrm{MC}}\ln2}{\mathcal{W}\left(\frac{1}{e}\left(\frac{\mathcal{C}h_k}{(q+\lambda)N_0}-1\right)\right)+1}, W^k_{\mathrm{MC}}\right), \forall \, k \in \Psi,  \label{eq4.22}\\
p_{k,k'}&=(W^k_{\mathrm{MC}}-w_k)\left[\frac{(1+\mu)\xi}{({q}+\lambda\xi)\ln2}-\frac{N_0}{g_{k,k'}}\right]^+, \forall \, k \in \Psi,  \label{eq4.25}\\
p_{n}&=B^n_{\mathrm{SC}}\left[\frac{(1+\mu)\xi}{({q}+\lambda\xi)\ln2}-\frac{N_0}{g_{n}}\right]^+, \forall \, n\in \mathcal{N},  \label{eq4.26}
\end{align}
where $[x]^+\triangleq \max\{x,0\}$ and $\mathcal{W}(x)$ is the Lambert $\mathcal{W}$ function \cite{guo2014joint},  i.e., $x=\mathcal{W}(x)e^{\mathcal{W}(x)}$. In addition, $\mathcal{C}= (1+\mu)\log_2\left(1+{\widetilde{p}}_{k,k'}\frac{{g_{k,k'}}}{N_0}\right)-\left(\frac{q}{\xi}+\lambda\right)\widetilde{p}_k$, and ${\widetilde{p}}_{k,k'}=\left[\frac{(1+\mu)\xi}{({q}+\lambda\xi)\ln2}-\frac{N_0}{g_{k,k'}}\right]^+$.

 \end{theorem}
 \begin{proof}
The optimal bandwidth and power allocation in (\ref{eq4.22})-(\ref{eq4.26}) can be derived directly by analyzing KKT conditions of problem (\ref{eq13c1}).
\end{proof}

From (\ref{eq4.22}),  it is easy to show that the bandwidth allocated to MU $k$ by the SC, i.e., $w_k$, increases with its minimum required data rate by the MC, $R^{k}_{\mathrm{MC}}$, while decreasing with its channel power gain, $h_k$.  This implies that the SC is able to seek more bandwidth from the MUs who require lower user data rates but are closer to the SC BS, which also coincides with the intuition discussed previously. Furthermore,
we also observe that the optimal transmit power allocations, $p_{k,k'}$ and $p_n$, follow the conventional multi-level water-filling structure due to different bandwidth allocations. In contrast, the optimal transmit power densities, $\frac{p_{k,k'}}{W^k_{\mathrm{MC}}-w_k}$ and $\frac{p_n}{B^n_{\mathrm{SC}}}$, follow the conventional single-level water-filling structure \cite{ng2012energy1}. 
{{\begin{table}[]
\caption{\small{Energy-Efficient Joint Bandwidth and Power Allocation Algorithm}} \label{tab:overall} \centering
\vspace{-0.5cm}
\begin{algorithm}[H]
 \caption{Energy-Efficient Joint Bandwidth and Power Allocation [Dinkelbach method]} 
  \begin{algorithmic}[1]
\STATE \normalsize {  {\bf{Initialize}} the maximum tolerance  $\epsilon\ll 1$ and set $q=1$ with given MU $k$ and SU $k'$;\\ 
\STATE            {\bf{repeat}}  \\
\STATE          ~~~Initialize ${\lambda}$, ${\mu}$; \\
\STATE        ~~~{\bf{repeat}}  \\
\STATE         ~~~~~Obtain $w_k, p_{k,k'}, p_n$ from  (\ref{eq4.22})-(\ref{eq4.26}); \\
\STATE         ~~~~~Obtain $b_{k, k'}, q_k$ from C2 and C3 in (10);\\
\STATE         ~~~~~~~Update dual variables ${\lambda}$ and ${\mu}$ by the ellipsoid method;\\
\STATE        ~~~{\bf{until}} ${\lambda}$ and ${\mu}$ converge;     \\
\STATE         ~~~Update $q = \frac{R_{\rm{\rm{tot}}}({S})}{P_{\rm{\rm{tot}}}({S})}$; \\
\STATE            {\bf{until}}   $\Big(R_{\rm{\rm{tot}}}({S})-qP_{\rm{\rm{tot}}}({S})\Big) \leq \epsilon$. }
\end{algorithmic}
\end{algorithm}
\end{table}}}

The commonly adopted ellipsoid method can be employed iteratively for updating $({\lambda}, {\mu})$ toward the optimal solution with guaranteed convergence \cite{ng2012energy1}.
 A discussion regarding the choice of the initial ellipsoid,  the updating of the ellipsoid, and the stopping criterion for the ellipsoid method can be found in \cite[Section \Rmnum{5}-B]{yu2006dual} and is thus omitted here for brevity.
 Due to the concavity of problem (\ref{eq13c1}), the iterative optimization between $(p_n, p_{k,k'}, w_k)$ and $({\lambda}, {\mu})$ is guaranteed to converge to the optimal solution of (\ref{eq13c1}). The details of the bandwidth and power allocation for a given MU selection are summarized in Algorithm $1$ in Table I.

\section{Energy-Efficient MU Selection}
In this section, we investigate the MU selection problem, i.e., to find the MU set $\Psi$ where $x_k=1,\forall\, k\in\Psi$. 
 \subsection{Trading EE}
  The trading EE of MU $k$, for $k\in \mathcal{K}$, is defined as the total data rate of MU $k$ brought for the SC over the total power consumed by the SC in the spectrum-power trading, i.e.,
\begin{align}
EE_{k}=\frac{b_{k,k'}\log_2\left(1+\frac{p_{k,k'}g_{k,k'}}{b_{k,k'}N_0}\right)}{ \frac{p_{k,k'}}{\xi}+\frac{q_k}{\xi}},
\end{align}
where the numerator, $b_{k,k'}\log_2\left(1+\frac{p_{k,k'}g_{k,k'}}{b_{k,k'}N_0}\right)$, is the additional data rate obtained by the SC via serving MU $k$ and the denominator, $\frac{p_{k,k'}}{\xi}+\frac{q_k}{\xi}$, is the total power consumed for both supporting SU $k'$ and meeting the QoS of MU $k$.  As a result, the trading EE is in fact an evaluation of an MU in terms of the power utilization efficiency and can be regarded as a marginal benefit of the SC in the spectrum-power trading.
Then, the trading EE maximization problem of MU $k$ can be formulated as
 \begin{align}\label{eq20}
\mathop {\text{max} }\limits_{\overset{p_{k,k'}, b_{k,k'},}{ q_k, w_k}} ~~~& EE_{k}=\frac{b_{k,k'}\log_2\left(1+\frac{p_{k,k'}g_{k,k'}}{b_{k,k'}N_0}\right)}{ \frac{p_{k,k'}}{\xi}+\frac{q_k}{\xi}}\nonumber \\
\text{s.t.} ~~~~~~~~
& \text{C2:}~~b_{k,k'}+w_k\leq W^k_{\mathrm{MC}},  \nonumber \\
&\text{C4:}~~w_{k}\log_2\left(1+\frac{q_{k}h_{k}}{w_{k}N_0}\right)\geq R^k_{\mathrm{MC}},  \nonumber \\
& \text{C7:}~~b_{k,k'}\geq  0, ~ w_{k}\geq  0.    
\end{align}
It is worth noting that problem (\ref{eq20}) can be regarded as a special case of problem (\ref{eq16}) where there is only one MU and one SU. Therefore, problem (\ref{eq20}) can be solved similarly by the algorithm proposed in Section III.

 \subsection{Trading EE based MU Selection}
 The key observation of the user trading EE is that both $b_{k,k'}\log_2\left(1+\frac{p_{k,k'}g_{k,k'}}{b_{k,k'}N_0}\right)$ and $\frac{p_{k,k'}}{\xi}+\frac{q_k}{\xi}$ will be removed from the numerator and the denominator of the objective function in problem (\ref{eq17}), respectively,  if MU $k$ is not served by the SC.
 With the user trading EE defined in Section V-A, we now investigate the MU selection conditions for different cases. Recall that $\Psi$ denotes an arbitrary set of  MUs that are scheduled by the SC, i.e., $\Psi\triangleq\{k\,|\,x_k=1, k\in \mathcal{K}\}$, and  $EE^*_{\Psi}$ denotes the maximum system EE of problem (\ref{eq17}), which can be obtained by Algorithm $1$ based on set $\Psi$. Then, we have the following theorem to facilitate the algorithm development.

 \begin{theorem}\label{scheduling}
 For any unscheduled MU $m$, i.e.,  $m\in \mathcal{K}, m\notin \Psi$:
 \begin{enumerate}
 \item in the absence of constraints C1 and C4 in problem (\ref{eq17}),  serving MU $m$ improves the EE of the SC \emph{if and only if} $EE^*_m> EE^*_{\Psi}$;
   \item in the absence of constraint C1  in problem (\ref{eq17}),  serving MU $m$ improves the EE of the SC \emph{if} $EE^*_m> EE^*_{\Psi}$;
   \item in the absence of constraint C4  in problem (\ref{eq17}),  serving MU $m$ improves the EE of the SC \emph{only if} $EE^*_m>EE^*_{\Psi}$.
 \end{enumerate}
\end{theorem}
\begin{proof}
  {Let $\mathcal{S^*}=\Big\{\{{p}^*_{n}\}, \{{p}^*_{k,k'}\}, \{{q}^*_k\}, \{{b}^*_{k,k'}\}, \{{w}^*_k\}\Big\}$ denote the optimal solution of problem (\ref{eq20}) and its corresponding user EE is denoted as $EE_k$. Let $\mathcal{\widehat{S}}=\Big\{\{{p}^*_{n}\}, \{\widehat{p}_{k,k'}\}, \{\widehat{q}_k\},$  $\{\widehat{b}_{k,k'}\}, \{\widehat{w}_k\}\Big\}$ and $\mathcal{\widetilde{S}}=\Big\{\{\widetilde{{p}}_{n}\}, \{\widetilde{p}_{k,k'}\}, \{\widetilde{q}_k\}, \{\widetilde{b}_{k,k'}\},\{\widetilde{w}_k\}\Big\}$ denote the optimal solutions of problem (\ref{eq16}) with $x_k=1$ for $k\in \Psi$ and  $k\in \Psi\bigcup \{m\}$, respectively,  where $m\notin \Psi$.} The corresponding system EEs are denoted as $EE^*_{\Psi}$ and $EE^*_{\Psi\bigcup \{m\}}$, respectively.
  Let $R(\widehat{S})\triangleq \sum_{n=1}^Nr^n_{\mathrm{SC}}(\widehat{p}_{n})+\sum_{k\neq m}r_{k,k'}(\widehat{b}_{k,k'},\widehat{p}_{k,k'})$ ,
  $P(\widehat{S})\triangleq \sum_{n=1}^N\frac{\widehat{p}_n}{\xi} + \sum_{k\neq m}\frac{\widehat{p}_{k,k'}}{\xi}+\sum_{k\neq m}\frac{\widehat{q}_k}{\xi}$,
  $R(\widetilde{S})\triangleq \sum_{n=1}^Nr^n_{\mathrm{SC}}(\widetilde{p}_{n})+\sum_{k\neq m}r_{k,k'}(\widetilde{b}_{k,k'},\widetilde{p}_{k,k'})$, and
  $P(\widetilde{S})\triangleq \sum_{n=1}^N\frac{\widetilde{p}_n}{\xi} + \sum_{k\neq m}\frac{\widetilde{p}_{k,k'}}{\xi}+\sum_{k\neq m}\frac{\widetilde{q}_k}{\xi}$.
Then, we have the following
\begin{align}\label{eq39}
EE^*_{\Psi\bigcup \{m\}}
&= \frac{R(\widetilde{S}) + r_{m,m'}(\widetilde{b}_{m,m'},\widetilde{p}_{m,m'})}
{P(\widetilde{S})+P_\mathrm{c} + \frac{\widetilde{p}_{m,m'}}{\xi}+ \frac{\widetilde{q}_m}{\xi}} \nonumber  \\
&\overset{(a)} \geq \frac{R(\widehat{S}) + r_{m,m'}({b}^*_{m,m'},{p}^*_{m,m'})}
{P(\widehat{S})+P_\mathrm{c} + \frac{{p}^*_{m,m'}}{\xi}+ \frac{{q}^*_m}{\xi}}  \nonumber  \\
&\overset{(b)} \geq \min\left\{\frac{R(\widehat{S}) }
{P(\widehat{S})+P_\mathrm{c}},
\frac{ r_{m,m'}({b}^*_{m,m'},{p}^*_{m,m'})}{\frac{p^*_{m,m'}}{\xi}+ \frac{q^*_m}{\xi}} \right\}    \nonumber \\
&= \min\left\{EE^*_{\Psi}, EE^*_{m}\right\},
\end{align}
where inequality $(a)$ holds due to the fact that $\widetilde{S}$ is the optimal solution of problem (\ref{eq16}) with $x_k=1$ for $k\in \Psi\bigcup \{m\}$. Inequality  $(b)$ holds due to  Lemma 1 and the equality ``='' holds only when $EE^*_{\Psi}= EE^*_{m}$. Thus, we can conclude $EE^*_{m} > EE^*_{\Psi}$ $\Longrightarrow$ $EE^*_{\Psi\bigcup \{m\}} > EE^*_{\Psi}$, which completes the proof of the ``if'' part. The ``only if" part can be proved similarly by exploiting the fractional structure as in (\ref{eq39}). Statements 2) and 3) can be readily obtained by analyzing $(a)$ and $(b)$ for the case of equality.
\end{proof}

Theorem \ref{scheduling} reveals the relationship between the inequality $EE^*_k> EE^*_{\Psi}$ and the MU selection under different constraints in problem  (\ref{eq17}).  
 Since constraints C1 and C4  may not be met with equalities simultaneously at the optimal solution in most cases, condition $EE^*_k> EE^*_{\Psi}$ is either sufficient or necessary for serving MU $k$ in practice.
  It is also interesting to mention that $EE^*_k> EE^*_{\Psi}$ has an important practical interpretation: the SC performs spectrum-power trading with MU $k$ when the trading EE is higher than the current system EE of the SC.
    In other words, the spectrum-power trading with this MU enables the SC to have a better utilization of the power. Otherwise, the spectrum-power trading is only beneficial to the MC and does not bring any benefit for the SC.

The main implication of Theorem \ref{scheduling} is that an MU with a higher user trading EE is more likely to be scheduled by the SC.
Based on this insight, a computationally efficient MU selection scheme is designed as follows. First, we sort all the MUs in descending order according to the user trading EE. Second, for MU $k$ satisfying the condition  $EE^*_k> EE^*_{\Psi}$ in Theorem \ref{scheduling}, we set $x_k=1$ and maximize the system EE in problem (\ref{eq17}) by Algorithm $1$. Third, by comparing the updated system EE with previous system EE where $x_k=0$ holds, we decide whether to schedule MU $k$.  The details of the MU selection procedure is summarized in Algorithm $2$ in Table II.

{{\begin{table}[]
\caption{\small{Energy-Efficient Spectrum-Power Trading Algorithm}} \label{tab:overall} \centering
\vspace{-0.5cm}
\begin{algorithm}[H]
 \caption{Energy-Efficient Spectrum-Power Trading} 
  \begin{algorithmic}[1]
\STATE \normalsize {  Obtain $EE_k$, $\forall\, k$, by solving problem (\ref{eq20}); \\
\STATE Sort all the MUs in descending order of trading EE, i.e., $EE_{1}^{*} > EE_{2}^{*}>,...,> EE_{K}^{*}$; \\
\STATE Set  $\Psi={\O}$ and obtain $EE_{\Psi}^{*}$ by Algorithm 1; \\  \STATE {\bf{for}}  $k=1:K$\\
\STATE       ~~~ Obtain $EE^*_{\Psi \bigcup \{k\}}$ by Algorithm 1; \\
\STATE      ~~~ {\bf{if}}    $EE^*_{\Psi \bigcup \{k\}} > EE_{\Psi}^{*}$    \\
\STATE        ~~~~~~ $\Psi={\Psi\bigcup \{k\}}$;           \\    \STATE      ~~~ {\bf{end}} \\
\STATE {\bf{end} }}
\end{algorithmic}
\end{algorithm}
\end{table}}}

The computational complexity of Algorithm $2$ can be evaluated as follows. First, the computational complexity for obtaining bandwidth and power allocation variables in Algorithm $1$ increases linearly with the number of MUs and the number of SUs, i.e., $\mathcal{O}(K+N)$.  Second, the computational complexities of the ellipsoid method for updating dual variables and the Dinkelbach method for updating $q$ are both independent of $K$ and $N$\cite{Boyd,ng2012energy1}. Finally, the complexity of performing the MU selection linearly increases with $K$. Therefore, the total computational complexity of Algorithm $2$ is $\mathcal{O}\big(K(K+N)\big)$.

\vspace{-0.2cm}
\section{Numerical Results }
\begin{table}[!t]
\centering
\caption{\label{table2} \small{Simulation Parameters.}} \label{parameters}
\renewcommand\arraystretch{1.1}
\begin{tabular}{|c|c|}
\hline
{Parameter} & {Description} \\
\hline

Maximum transmit power of the SC, $P^{\mathrm{SC}}_{\max}$&      $30$  dBm       \\
\hline
Licensed bandwidth of each MU, $W^k_{\mathrm{MC}}$  &      $240$  kHz\\
\hline
Licensed bandwidth of each SU, $B^{n}_{\mathrm{SC}}$ &      $180$  kHz\\
\hline
Static circuit power of the SC, $P_\mathrm{c}$ &    $2$ W \\
\hline
Power spectral density of thermal noise  &    $-174$ dBm/Hz \\
\hline
Power amplifier efficiency, $\xi$ &    $0.38$    \\
\hline
Path loss model  &    $128.1 +37.6\log_{10}d$    \\
\hline
Lognormal shadowing  &    $8$ dB    \\
\hline
Penetration loss  &    $20$ dB    \\
\hline
Fading distribution  &    Rayleigh fading    \\
\hline
\end{tabular}
\end{table}

 We consider a two-tier heterogeneous network where there exists an MC and an SC with the coverage radii of $500$ m and  $50$ m, respectively \cite{ngo2014joint}. Five SUs are uniformly distributed within the coverage of the SC BS while five MUs are uniformly distributed within the distances of $[20 \,200]$ m away from the SC BS.
  The distance between the SC BS and the MC BS is set to $500$ m.  Without loss of generality, we assume that all SUs and MUs have identical parameters. Unless specified otherwise, the important parameters are listed in Table \ref{parameters} and $R_{\mathrm{SC}}$ and $R^k_{\mathrm{MC}}$ are set to be $1000$ kbits and $700$ kbits, respectively.
\subsection{Convergence of Algorithm $2$}
 \begin{figure}[!t]
\centering
\includegraphics[width=3.5in]{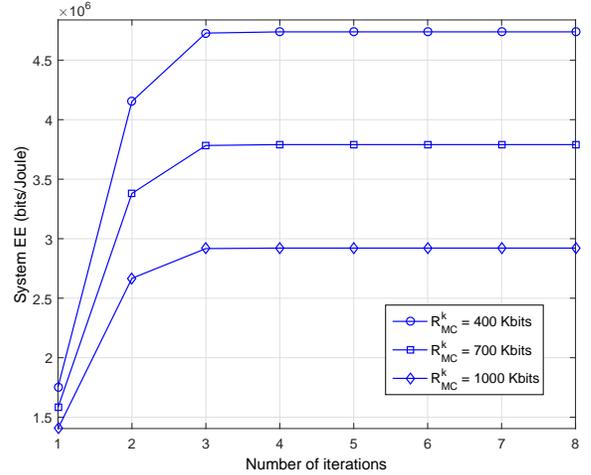}  
\caption{ System EE (bits/Joule) versus the number of  iterations in outer-layer of Algorithm $1$ for different  $R_{\mathrm{MC}}^k$. }\label{fig1}
\end{figure}
As mentioned, Algorithm $2$  is composed of two steps: MU selection and then joint bandwidth and power allocation. It has been shown  in Section V that the MU selection scheme is mainly based on the linear search of the trading EE order and is thus guaranteed to converge within at most $K$ iterations. As a result,  we only need to show the convergence of joint bandwidth and power allocation algorithm in Section IV, i.e., Algorithm $1$. Since the MU selection does not affect the convergence of Algorithm $1$, we set $x_k=1$, $\forall\,k$, to study the convergence.
Figure \ref{fig1} depicts the achieved system EE versus the number of  iterations of Dinkelbach method in Algorithm $1$. As can be observed, at most six iterations are needed on average to reach the optimal solution of the outer-layer problem. In addition,  the Lagrangian duality approach for the joint bandwidth and power allocation in the inner-layer problem also converges to the optimal solution due to the convexity of the problem (\ref{eq13c1})  \cite{Boyd}. Therefore,  Algorithm $2$ is guaranteed to converge.

\begin{figure}[!t]
\centering
\includegraphics[width=3.5in]{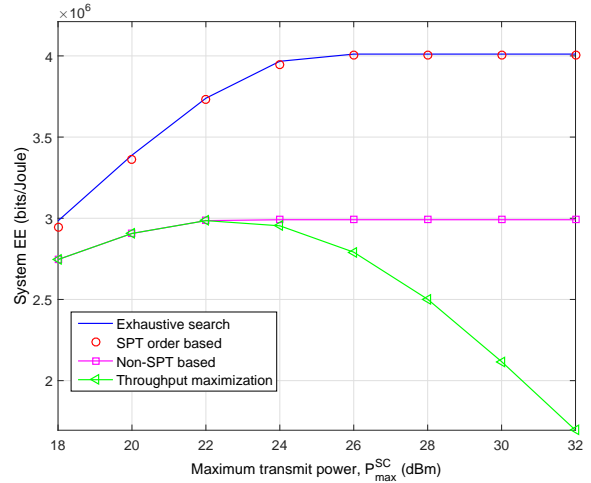}
\caption{System EE (bits/Joule) versus the maximum allowed transmit power of the SC.}\label{fig2}
\end{figure}



\subsection{System EE versus Maximum Transmit Power of SC}
In Figure \ref{fig2}, we compare the achieved system EE  of the following schemes: 1) Exhaustive search \cite{Boyd}; 2) spectrum-power trading (SPT) order based: Algorithm $2$ in Section V; 3) Non-SPT based: the EE maximization without spectrum-power trading \cite{ng2012energy2}; 4) Throughput Maximization: conventional spectral efficiency maximization \cite{ngo2014joint}.  It is observed that the proposed Algorithm $2$ achieves a near-optimal performance and outperforms all other suboptimal schemes, which demonstrates the effectiveness of the proposed scheme.
We also observe that the EEs of the SPT order based scheme and the non-SPT based scheme first increases and then remain constants as $P^{\mathrm{SC}}_{\max}$ increases. In contrast, the EE of the throughput maximization scheme first increases and then decreases with increasing $P^{\mathrm{SC}}_{\max}$, which is due to its greedy use of the transmit power. In addition, it is also seen that the performance gap between the  SPT order based scheme and the non-SPT based scheme first increases and then approaches a constant. This is because when the transmit power of the SC is limited, e.g.  $P^{\mathrm{SC}}_{\max}\leq20$ dBm, the SC may not have sufficient transmit power to serve many MUs and thereby the spectrum-power trading is less likely to be realized, which in return limits its own performance improvement. As $P^{\mathrm{SC}}_{\max}$ increases, compared to the non-SPT based scheme, the SC not only has more transmit power to improve its EE via serving its own SUs, but also has more transmit power to obtain additional bandwidth from the MC via spectrum-power trading, which thereby strengthens the effect of performance improvement. Finally, when all the `good' MUs with higher trading EE are being scheduled by the SC, then the system EE improves with $P^{\mathrm{SC}}_{\max}$ with diminishing return and eventually approaches a constant due to the same reason as that of the non-SPT based scheme.
\begin{figure}[!t]
\centering
\includegraphics[width=3.5in]{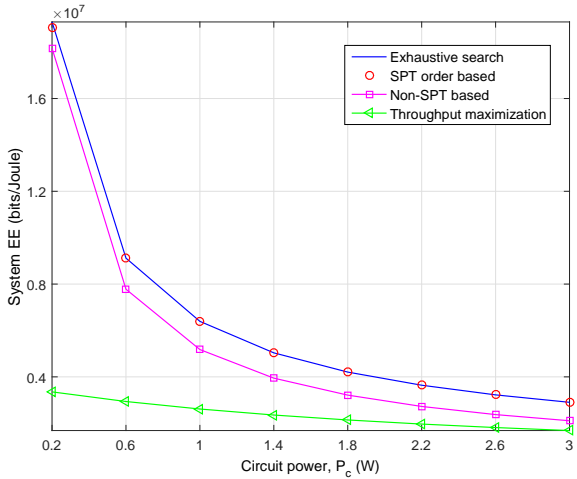}
\caption{System EE versus the circuit power of the SC.}\label{fig3}
\end{figure}
\vspace{-0.3cm}
  \subsection{System EE versus Circuit Power of  SC }
  Figure \ref{fig3} illustrates the performance of all the schemes as a function of the circuit power consumption of the SC. We can observe that the system EE of the all the schemes decreases with an increasing $P_\mathrm{c}$ since the circuit power consumption is always detrimental to the system EE. Also, the proposed Algorithm $2$ performs almost the same as the exhaustive search in all the considered scenarios.
  In addition, the performance gap between the non-SPT scheme and the throughput maximization scheme decreases with an increasing $P_{{\mathrm{c}}}$. This is because
 as  $P_{c}$ increases, the circuit power consumption dominates the total power consumption rather than the transmit power consumption. Thus, improving the system EE is almost equivalent to improving the system data rate, which only results in marginal performance gap \cite{ramamonjison2015energy}.

 However, it is interesting to note that the performance gap between the SPT order based scheme and the non-SPT based scheme does not decrease but increases when $P_\mathrm{c}$ is in a relatively small regime, e.g. $P_\mathrm{c} \in [0.2 ~ 1]$ W. This is because when $P_\mathrm{c}$ is very small, the SC system itself enjoys a high system EE which leaves it a less incentive to perform spectrum-power trading with the MC. Thus, the system EE of the SPT order based scheme  decreases with a similar slope with that of the non-SPT based scheme. As $P_\mathrm{c}$ increases, the system EE of the SC further decreases, which would motivate the SC to perform spectrum-power trading. As a result, the performance degradation caused by an increasing $P_\mathrm{c}$ is relieved by the spectrum-power trading in the proposed scheme, which thereby yields an increased performance gap between these two schemes in small $P_\mathrm{c}$ regime.  Furthermore, when $P_\mathrm{c}$ is sufficiently large such that all the `good' MUs are being selected, the performance gap between these two schemes decreases again due to the the domination of the circuit power in the total power consumption.

\section{Conclusions}
In this paper, we investigated the spectrum-power trading between an SC and an MC to improve the system EE of the SC. Specifically, MU selection, bandwidth allocation, and power allocation were jointly optimized while guaranteeing the QoS of both networks.
 {Simulation results showed that the proposed algorithm obtained a close-to-optimal performance and  also demonstrated the performance gains achieved by the proposed spectrum-power trading scheme for both the SC and the MC.

\bibliographystyle{IEEEtran}
\bibliography{IEEEabrv,mybib}

\end{document}